\documentclass[12pt]{amsart}

\usepackage[hmargin=2.5cm,vmargin=2.5cm]{geometry}
\usepackage{graphicx}
\usepackage{color}

\newtheorem{theorem}{Theorem}
\newtheorem{lemma}{Lemma}
\newtheorem{corollary}{Corollary}

\theoremstyle{remark}
\newtheorem{remark}{Remark}

\newcommand{\cH}{\mathcal{H}}
\newcommand{\R}{\mathbb{R}}
\newcommand{\C}{\mathbb{C}}

\DeclareMathOperator{\ind}{ind}

\numberwithin{subsection}{section}

\begin{document}

\title[Magnetic perturbation and nodal count: quantum graphs]{Stability of eigenvalues
  of quantum graphs with respect to magnetic perturbation and the
  nodal count of the eigenfunctions}

\author{Gregory Berkolaiko and Tracy Weyand}
\address{Department of Mathematics, Texas A\&M University, College
  Station, TX 77843-3368, USA}

  \begin{abstract}
  We prove an analogue of the magnetic nodal theorem on quantum
  graphs: the number of zeros $\phi$ of the $n$-th eigenfunction of the
  Schr\"odinger operator on a quantum graph is related to the
  stability of the $n$-th eigenvalue of the perturbation of the
  operator by magnetic potential.  More precisely, we consider the
  $n$-th eigenvalue as a function of the magnetic perturbation and
  show that its Morse index at zero magnetic field is equal to $\phi -
  (n-1)$.
  \end{abstract}


\maketitle


\section{Introduction}

A quantum graph is a metric graph equipped with a self-adjoint
differential ``Hamiltonian'' operator (usually of Schr\"odinger type)
defined on the edges and matching conditions specified at the
vertices.  Graph models in general, and quantum graphs in particular,
have long been used as a simpler setting to study complicated
phenomena.  We refer the interested reader to the reviews
\cite{Kuc_wrm02,GnuSmi_ap06,Kuc_incol08}, collections of papers
\cite{BerCarFulKuc_eds06,ExnKeaKuc_eds08}, and the recent monograph
\cite{BK_graphs} for an introduction to quantum graphs and their
applications.

Quantum graphs have been especially fruitful models for studying the
properties of zeros of the eigenfunctions
\cite{BanOreSmi_pspm08,BanBerSmi_ahp12}.  Of particular interest is
the relationship between the sequential number of the eigenfunction
and the number of its zeros, which we will refer to as the \emph{nodal
  point count}.  It was on quantum graphs that the relationship
between the stability of the nodal partition of an eigenfunction and
its nodal deficiency was first discussed \cite{Banetal_cmp12}.  Since
then, the result has been extended to discrete graphs
\cite{BerRazSmi_jpa12} and bounded domains in $\R^d$
\cite{BerKucSmi_gafa12}.

In a similar-spirited development, it has been discovered that the
nodal point count on discrete graphs is connected to the stability of
the eigenvalue with respect to a perturbation by a magnetic field
\cite{Ber_prep11} (see also \cite{Ver_prep12} for an alternative
proof).  While the magnetic result drew inspiration from the
developments for nodal partitions, the relationship between the two
results was very implicit in the original proof \cite{Ber_prep11} and
was not at all relevant in the proof of \cite{Ver_prep12}.

The purpose of the current paper is three-fold.  We prove an analogue
of the magnetic theorem of \cite{Ber_prep11} on quantum graphs.  This
is done by establishing a clear and explicit link between magnetic
perturbation and the perturbation of the nodal partition.  Along the
way, we remove some superfluous (and troublesome) assumptions from the
nodal partitions theorem of \cite{Banetal_cmp12}.

A proof of the magnetic theorem on the simplest of quantum graphs, a
circle, has already been found in \cite{Ver_prep12}.  This proof uses
the explicitly available nodal point count and thus is impossible to generalize
to any non-trivial graph.  However, we do acknowledge drawing
inspiration (in particular, in the use of Wronskian) from the work of
\cite{Ver_prep12}.

Finally, we would like to mention that the main result of the present
paper has already been used by R.~Band to prove an elegant ``inverse
nodal theorem'' on quantum graphs \cite{Ban_prep12}, which appears in
this same volume.


\section{Main results}

We start by defining the quantum graph, following the notational
conventions of \cite{BK_graphs}.  We also refer the reader to
\cite{BK_graphs} for the proofs of all background results used in this
section.

Let $\Gamma$ be a compact metric graph with vertex set $V$ and edge
set $E$.  Let $\widetilde{H}^k(\Gamma,\C)$ be the space of all complex-valued
functions that are in the Sobolev space $H^k(e)$ for each edge, or in other
words
\begin{equation*}
  \widetilde{H}^k(\Gamma,\C) = \oplus_{e \in E} H^k(e),
\end{equation*}
while $\widetilde{H}^k(\Gamma,\R)$ will denote the space of all real-valued
functions that are in the Sobolev space $H^k(e)$ for each edge.
We define $\widetilde{L}^k(\Gamma,\C)$ and
$\widetilde{L}^k(\Gamma,\R)$ similarly.  Consider the Schr\"odinger
operator with electric potential $q:\Gamma\rightarrow \R$
defined by
\begin{equation*}
  H^0: f \mapsto -\frac{d^2f}{dx^2} + q f,
\end{equation*}
acting on the functions from $\widetilde{H}^2(\Gamma,\C)$ satisfying the
$\delta$-type boundary conditions
\begin{equation}\label{eq:vconditions}
  \left\{\begin{array}{l}
      f(x) \mbox{ is continuous at v}, \\
      \sum_{e\in {E_v}} \frac{df}{dx_e}(v) = \chi_vf(v), \qquad
      \chi_v\in \R.
    \end{array}\right.
\end{equation}
Here the potential $q(x)$ is assumed to be piecewise continuous.
The set $E_v$ is the set of edges joined at the vertex $v$; by
convention, each derivative at a vertex is taken into the
corresponding edge. We denote by $x_e$ the local coordinate on edge $e$.

On vertices of degree one, we also allow the Dirichlet condition
$f(v)=0$, which is formally equivalent to $\chi_v = \infty$.  In this
case, we do not count the Dirichlet vertex as a zero, neither when
specifying restrictions on the eigenfunction nor when counting its
zeros.

The operator $H^0$ is self-adjoint, bounded from below, and has a
discrete set of eigenvalues that can be ordered as
\begin{equation*}
  \lambda_1 \leq \lambda_2 \leq \ldots \leq \lambda_n \leq \ldots \hspace{.1in}.
\end{equation*}

The magnetic Schr\"odinger operator on $\Gamma$ is given by
\begin{equation*}
  H^A(\Gamma): f \mapsto -\left(\frac{d}{dx} - iA(x)\right)^2f + qf,
  \qquad f \in \widetilde{H}^2(\Gamma,\C),
\end{equation*}
where the one-form $A(x)$ is the magnetic potential (namely, the sign of
$A(x)$ changes with the orientation of the edge).  The $\delta$-type
boundary conditions are now modified to
\begin{equation*}
  \left\{\begin{array}{l}
      f(x) \mbox{ is continuous at v}, \\
      \sum_{e\in {E_v}} \left(\frac{df}{dx_e}(v) - iA(v)f(v)\right) = \chi_vf(v), \qquad
      \chi_v\in \R.
    \end{array}\right.
\end{equation*}

Let $\beta = |E| - |V| + 1$ be the first Betti number of the graph
$\Gamma$, i.e.\ the rank of the fundamental group of the graph.
Informally speaking, $\beta$ is the number of ``independent'' cycles
on the graph.  Up to a change of gauge, a magnetic field
on a graph is fully specified by $\beta$ fluxes $\alpha_1,
\alpha_2, \ldots, \alpha_\beta$, defined as
\begin{equation*}
  \alpha_j = \oint_{\sigma_j} A(x) dx \mod 2\pi,
\end{equation*}
where $\{\sigma_j\}$ is a set of generators of the fundamental group.
In other words, magnetic Schr\"odinger operators with different
magnetic potentials $A(x)$, but the same fluxes
$(\alpha_1,\ldots, \alpha_\beta)$, are unitarily equivalent.  Therefore,
the eigenvalues $\lambda_n(H^A)$ can be viewed as functions of
$\boldsymbol{\alpha} = (\alpha_1,\ldots, \alpha_\beta)$.

In this paper, we will prove the following main result:
\begin{theorem}
  \label{thm:main_mag}
  Let $\psi$ be the eigenfunction of $H^0$ that corresponds to a
  simple eigenvalue $\lambda = \lambda_n(H^0)$.  We assume that $\psi$
  is non-zero on vertices of the graph.  We denote by $\phi$ the number
  of internal zeros of $\psi$ on $\Gamma$.

  Consider the perturbation $H^A$ of the operator $H^0$ by a magnetic
  field $A$ with fluxes $\boldsymbol\alpha = (\alpha_1,\ldots,
  \alpha_\beta)$.  Then $\boldsymbol{\alpha} = (0,\ldots, 0)$ is a
  non-degenerate critical point of the function
  $\lambda_n(\boldsymbol{\alpha}) := \lambda_n(H^A)$ and its Morse index is
  equal to the nodal surplus $\phi - (n - 1)$.
\end{theorem}

To prove this theorem, we will study the eigenvalues of a tree which
is obtained from $\Gamma$ by cutting its cycles and introducing
parameter-dependent $\delta$-type conditions on the newly formed
vertices.  The parameter-dependent eigenvalues of the cut tree will be
related to the magnetic eigenvalues $\lambda_n(\boldsymbol{\alpha})$ via an
intermediate operator, which can be viewed as a magnetic Schr\"odinger
operator with imaginary magnetic field.


\section{Cutting the graph}  \label{sec:cutting}

A spanning tree of a graph $\Gamma = \{V,E\}$ is a tree composed of all the vertices
$V$ and a subset of the edges $E$ that connects all of the vertices
but forms no cycles.  Choose a spanning tree of the graph $\Gamma$ and let $C$ be the set of
edges that is complementary to the chosen tree.  It is a classical
result that $|C|=\beta$ independently of the chosen spanning tree.  On
each of the edges from $C$, we choose an arbitrary point $c_j$.  If we
cut the graph at all points $\{c_j\}$, each point will give rise to
two new vertices which will be denoted $c_j^+$ and $c_j^-$, see
Figure~\ref{fig:cutting}.  The new graph is a tree and will be denoted
$T$; it can be viewed as a metric analogue of the notion of the spanning tree.  By
specifying different vertex conditions on the new vertices $c_j^\pm$,
we will obtain several parameter-dependent families of quantum trees.

The first family makes precise the above discussion of equivalences
among the magnetic operators $H^A$.  This easy result can be
found, for example, in \cite{BK_graphs,KosSch_cmp03,Rue_prep11}.

\begin{lemma} \label{lem:unit_equiv} The operator $H^A$ is unitarily
  equivalent to the operator $H^{\boldsymbol{\alpha}}: \widetilde{H^2}(T,\C) \rightarrow
  \tilde{L}^2(T,\C)$, defined as $-\frac{d^2}{dx^2} + q(x)$ on every
  edge, with the same vertex conditions on the vertices of $T$
  inherited from $\Gamma$ (see equation \eqref{eq:vconditions}), and
  the Robin conditions
  \begin{equation} \label{eq:phase_jump_cond}
    \begin{split}
      f(c_j^+) &= e^{i \alpha_j} f(c_j^-), \\
      f'(c_j^+) &= -e^{i \alpha_j} f'(c_j^-),
    \end{split}
 \end{equation}
 at the new vertices, where the phases $\alpha_j$ are determined by
 \begin{equation}  \label{eq:phase_from_A}
   \alpha_j = \int_{c_j^-}^{c_j^+} A(x) dx  \mod 2\pi,
 \end{equation}
 with the integral taken over the unique path on $\Gamma$ that does
 not pass through any other cut points $c_i$.
\end{lemma}

\begin{remark}
  The minus sign in the second equation of \eqref{eq:phase_jump_cond}
  is due to the fact that at the vertices $c_j^-$ and $c_j^+$ of the
  tree $T$, the derivatives are taken into the edges.
\end{remark}

Henceforth, by $H^{\boldsymbol{\alpha}}$ we will also denote the equivalence
class of operators $H^A$ on $\Gamma$ that are unitarily equivalent to
$H^{\boldsymbol{\alpha}}$.  Since we will be solely interested in the
eigenvalues of $H^A$, this constitutes only a slight abuse of
notation.

To use what we know about zeros of eigenfunctions on \emph{trees} we
need local conditions (unlike those in (\ref{eq:phase_jump_cond})) at
the cut vertices $c_j^\pm$.

Starting with the graph $T$, we define a family of operators
$H_{\boldsymbol\gamma}$ where $\boldsymbol{\gamma} = \left(\gamma_1, \ldots,
  \gamma_\beta \right) \in \R^\beta$.  The operator $H_{\boldsymbol{\gamma}}$ acts as
$-\frac{d^2}{dx^2} + q(x)$ on $f\in \widetilde{H^2}(T,\R)$ that satisfy
the conditions
\begin{equation}  \label{eq:cut_cond}
  \begin{split}
    f'(c_j^+) &= \gamma_j f(c_j^+),\\
    f'(c_j^-) &= -\gamma_j f(c_j^-),
  \end{split}
\end{equation}
at the cut points, together with whatever conditions were imposed on
the vertices of the original graph $\Gamma$.

\begin{figure}[t]
  \centering
  \includegraphics{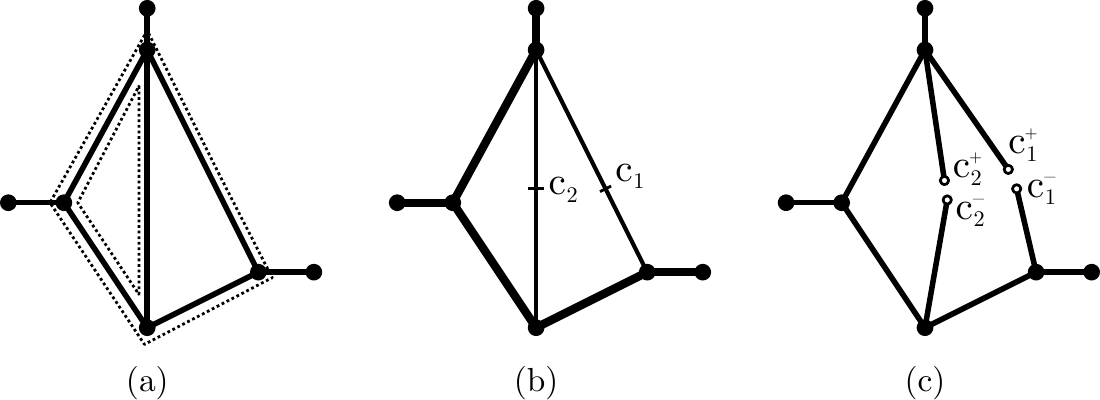}
  \caption{A graph with a choice of generators of the fundamental
    group (a); a choice of the spanning tree (thicker edges) and the
    cut points (b); the cut tree $T$ with new vertices $c_j^\pm$ (c).}
  \label{fig:cutting}
\end{figure}

We consider the $m$-th eigenvalue $\lambda_m(H_{\boldsymbol{\gamma}})$ as a
function of $\boldsymbol{\gamma}$.  We will prove that each eigenfunction
$\psi_n$ of $H^0$ gives rise to a critical point of the function
$\lambda_m(H_{\boldsymbol\gamma})$ (for a suitable $m$) and find the Morse index
of this critical point.  This problem was first considered in
\cite{Banetal_cmp12} to study the partitions of the graph $\Gamma$.
The results of \cite{Banetal_cmp12} contained an \emph{a priori} condition of
non-degeneracy of the critical point which rendered them unsuitable
for the task of proving Theorem~\ref{thm:main_mag}.
Theorem~\ref{thm:main_cut} below removes this extraneous condition and
generalizes the results of \cite{Banetal_cmp12}.  We discuss the
connection to \cite{Banetal_cmp12} in more detail in
section~\ref{sec:partition}.

Let $\lambda_n$ be a simple eigenvalue of $H^0$ and $\psi$ be the
corresponding eigenfunction.  Assume that the function $\psi$ is
non-zero at the vertices of the graph and at the cut points $c_j$
(moving the cut points if necessary).  Since $\psi \in
\widetilde{H}^2(\Gamma,\R)$, it is continuous and has continuous
derivatives.  Considering $\psi$ as a function on $T$, at every cut
point $c^\pm_j$ we have
\begin{equation*}
  \psi(c_j^+) = \psi(c_j^-)
  \qquad\mbox{and}\qquad
  \psi'(c_j^+) = -\psi'(c_j^-).
\end{equation*}

\begin{theorem}\label{thm:main_cut}
  Let $\psi$ be the eigenfunction of $H^0$ that corresponds to a
  simple eigenvalue $\lambda_n(H^0)$.  We assume that $\psi$ is
  non-zero on vertices of the graph.  We denote by $\phi$ the number of
  internal zeros of $\psi$ on $\Gamma$.  Define
  \begin{equation}  \label{eq:gamma_tilde_def}
    \tilde{\gamma}_j := \frac{\psi'(c_j^+)}{\psi(c_j^+)} = -\frac{\psi'(c_j^-)}{\psi(c_j^-)}
  \end{equation}
  and let $\widetilde{\boldsymbol\gamma} =
  (\tilde{\gamma}_1,\ldots,\tilde{\gamma}_\beta)$.
  Consider the eigenvalues of $H_{\boldsymbol{\gamma}}$ as functions
  $\lambda_n(H_{\boldsymbol{\gamma}})$ of $\boldsymbol{\gamma}$.  Then
  \begin{enumerate}
  \item \label{item:same_eig}
    $\lambda_{\phi+1}(H_{\boldsymbol{\gamma}})\Big|_{\boldsymbol{\gamma}=\widetilde{\boldsymbol\gamma}}
    = \lambda_n(H^0)$ where $\phi$ is the number of zeros of $\psi$
    on $\Gamma$,
  \item $\boldsymbol{\gamma} = \widetilde{\boldsymbol\gamma}$ is a non-degenerate
    critical point of the function
    $\lambda_{\phi+1}(H_{\boldsymbol{\gamma}})$, and
  \item the Morse index of the critical point $\boldsymbol{\gamma} = \widetilde{\boldsymbol\gamma}$
    is equal to $n-1+\beta - \phi$.
  \end{enumerate}

\end{theorem}


\section{Proof of Theorem \ref{thm:main_cut}}
\label{sec:proof_main_cut}

Before we prove Theorem \ref{thm:main_cut}, we
collect some preliminary results in subsections \ref{sec:qform}
through \ref{sec:reduction}.


\subsection{Quadratic form of $H_{\boldsymbol{\gamma}}$}\label{sec:qform}

The quadratic form of the operator $H^0$ on the graph $\Gamma$ is
 \begin{equation}\label{eq:qfh}
   h[f] = \sum_{e \in E(\Gamma)} \int_e \left(f'(x)^2 + q(x)f^2(x)\right)\,
   dx
   + \sum_{v \in V(\Gamma)} \chi_vf^2(v)
 \end{equation}
with the domain
\begin{equation*}
  D = \{f \in \widetilde{H}^1(\Gamma,\R): f \mbox{ is
    continuous at all vertices of $\Gamma$}\}.
\end{equation*}
The Dirichlet conditions, if any, are also imposed on the domain $D$.
For $\boldsymbol{\gamma} \in \R^\beta$, the quadratic form of the
operator $H_{\boldsymbol{\gamma}}$ acting on the tree $T$ is formally
 \begin{equation}\label{eq:qfhgamma}
   h_{\boldsymbol{\gamma}}[f] = h[f] + \sum_{j = 1}^\beta \gamma_j \left(f^2(c_j^+) - f^2(c_j^-)\right)
 \end{equation}
with the domain
\begin{equation*}
  D_{\boldsymbol\gamma} = \{f \in \widetilde{H}^1(T,\R): f \mbox{ is
    continuous at all vertices of the tree $T$}\}.
\end{equation*}
Importantly, the domain $D_{\boldsymbol\gamma}$ is larger than $D$ since we no longer impose
continuity at the cut points $c_j$ on the functions from $D_{\boldsymbol\gamma}$.
More precisely, we can represent the domain $D$ as
\begin{equation*}
  D = \{f \in D_{\boldsymbol\gamma}: f(c_j^+) = f(c_j^-)\}.
\end{equation*}
Note that the domain $D_{\boldsymbol\gamma}$ is independent of the actual value of
$\boldsymbol{\gamma}$.

\begin{remark}\label{rem:functiondecomp}
  Observe that any function $f \in D_{\boldsymbol\gamma}$ can be written as
  \begin{equation*}
    f = f_0 + \sum_{j=1}^\beta s_j\rho_j
  \end{equation*}
  where $f_0 \in D$, $s_j \in \R$ and $\rho_j \in
  D_{\boldsymbol\gamma}$.  Moreover, we require that $\rho_j$ have a
  jump at $c_j$, but be continuous at all other cut points $c_k$,
  $k\neq j$ (i.e. each $\rho_j$ represents one jump of the function $f$).
  In particular, for a given $\lambda$, we will use the family of
  functions $\rho_{j,\lambda}$ that satisfy
  \begin{equation*}
    -\rho_{j,\lambda}''(x) + (q(x) - \lambda)\rho_{j,\lambda}(x)= 0
  \end{equation*}
  on every edge, the $\delta$-type conditions (\ref{eq:vconditions})
  at the vertices of $\Gamma$, and the following conditions at the cut points:
 \begin{align}
   &\rho_{j,\lambda}(c_j^-) = 0,&
   &\mbox{and}& \qquad
   &\rho_{j,\lambda}(c_j^+) = 1 \\
   \label{eq:no_cut}
   &\rho_{j,\lambda}(c_k^+) = \rho_{j,\lambda}(c_k^-),&
   &\mbox{and}& \qquad
   &\rho_{j,\lambda}'(c_k^+) = -\rho_{j,\lambda}'(c_k^-)
   \hspace{.2in} \forall\hspace{.05in} k \neq j.
  \end{align}
  Note that condition (\ref{eq:no_cut}) essentially glues these cut points back together.

  Existence and uniqueness of the functions satisfying the above
  conditions is assured (see, for example \cite{BK_graphs}, section
  3.5.2) provided $\lambda$ stays away from the Dirichlet spectrum
  $\rho_{j,\lambda}(c_j^+) = 0$.  Since we are interested in $\lambda$
  close to the eignevalue $\lambda_n(H^0)$ of the uncut graph,
  we check that $\lambda_n(H^0)$ does not belong to the Dirichlet
  spectrum described above.  The corresponding ``Dirichlet graph'' can
  be viewed as the uncut graph with an extra Dirichlet condition
  imposed at the vertex $c_j$ (in place of the Neumann condition
  effectively imposed there by $H^0$).  However, by a simple extension
  of the interlacing theorem of \cite{BK_graphs,BerKuc_incol12} (see
  Lemma~\ref{lem:inequality} below for a precise formulation), one can
  see that the interlacing between Neumann and Dirichlet eigenvaules
  is strict since $\lambda_n(H^0)$ is assumed to be simple and the
  corresponding eigenfunction $\psi$ is non-zero at the cut points.
\end{remark}


\subsection{Properties of Wronskian on graphs}

It will be important to relate the values of the derivatives of the
functions $\rho_{j,\lambda}$ at the cut points $c_j^\pm$.  We will do
this using the Wronskian.  Therefore, in this subsection, we
investigate the properties of the Wronskian on graphs.  We will do this
for the most general self-adjoint vertex conditions on the graph $\Gamma$.

Given any two functions $f_1, f_2 \in D$ that satisfy the differential
equation $-f''(x) + (q(x) - \lambda)f(x) = 0$, we know by Abel's
formula that the Wronskian of $f_1$ and $f_2$ is constant on any
interval, or in particular, on any edge.  Observe that the
Wronskian is a one-form, that is its sign depends on direction.
We will now show that the total sum of Wronskians at any vertex with
self-adjoint conditions is zero (all Wronskians must be taken in the
outgoing direction).

\begin{lemma}\label{sum}
  Let $\Gamma$ be a graph and let $f_1, f_2 \in
  \widetilde{H}^2(\Gamma,\R)$ be two functions that satisfy
  the differential equation $-f''(x) + (q(x) - \lambda)f(x) = 0$
  and real self-adjoint vertex conditions.  Then
  \begin{equation*}
    \sum_{e \in E_v} W_e(f_1,f_2)(v) = 0,
  \end{equation*}
  where $E_v$ denotes the set of all edges attached to the vertex $v$
  and each Wronskian $W_e(f_1,f_2)$ is taken outward.
\end{lemma}

\begin{proof}
  We denote the self-adjoint operator acting as $-\frac{d^2}{dx^2} + q(x)$ by $H$.
  Define a smooth compactly supported function $\zeta$ on $\Gamma$ such that
  $\zeta \equiv 1$ in a neighborhood of the vertex $v$ and is zero at all other
  vertices of $\Gamma$.  For the sake of convenience, we denote $\zeta f_j$ by $g_j$.
  Then using the self-adjointness of $H$ and integrating by parts, we obtain
  \begin{align*}
    0 &= \langle Hg_1,g_2\rangle - \langle g_1,Hg_2\rangle \\
    &= \sum_{e \in E}  \left.\left(g_1(x)g_2'(x) -
        g_1'(x)g_2(x)\right)\right|^{L_e}_0
    + \int_e \left(g_1'(x)g_2'(x) - g_2'(x)g_1'(x)\right)\,dx\\
    &= \sum_{e \in E_v} W_e(g_1,g_2)(v) = \sum_{e \in E_v} W_e(f_1,f_2)(v)
  \end{align*}
  since $g_j = f_j$ near vertex $v$ and $g_j$ are zero near all other
  vertices.
\end{proof}

\begin{lemma}\label{lem:wronskian_cutpoints}
  Let $a$ and $b$ be two leaves (i.e., vertices of degree one) of a graph $\Gamma$.  Let $f_1$ and
  $f_2$ be two solutions of $-f''(x) +(q(x) - \lambda)f(x) = 0$ on
  $\Gamma$ that satisfy the same self-adjoint vertex conditions at all
  vertices except $a$ and $b$.  Then $W(f_1,f_2)(a) = -
  W(f_1,f_2)(b)$.
\end{lemma}

\begin{proof}
  In graph theory, a \emph{flow} $\eta$ between two vertices $a$ and
  $b$ is defined as a non-negative function on the edges of a directed
  graph $\Gamma$ that satisfies Kirchhoff's current conservation
  condition at every vertex other than $a$ or $b$: the total current
  flowing into a vertex must equal the total current flowing out of it
  (see Figure \ref{fig:flow} for an example).  Given a flow $\eta$
  between $a$ and $b$, it is a standard result of graph theory that
  the total current flowing into $b$ is equal to the total current
  flowing out of $a$ \cite{Bol_GTM98}.

  We interpret the Wronskian as a flow by assigning directions to the
  edges of $\Gamma$ so that the Wronskian is always positive.  The
  current conservation condition is then equivalent to Lemma
  \ref{sum}.  Therefore, the flow into $b$ equals the flow out of $a$
  so $W(f_1,f_2)(a) = - W(f_1,f_2)(b)$.
\end{proof}

\begin{figure}
  \centering
  \includegraphics{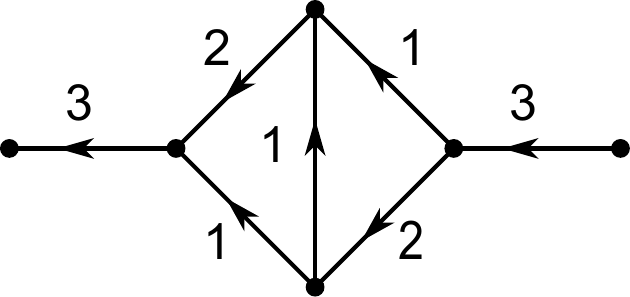}
  \caption{Let $\Gamma$ be the directed graph above and let $f$ be the function whose (constant) value on every edge is given by the number near each arrow.  One can easily check that the function $f$ constitutes a flow on $\Gamma$.}
  \label{fig:flow}
\end{figure}

\subsection{Morse index with Lagrange multipliers}
\label{sec:MorseLagrange}

In the proof of Theorem \ref{thm:main_cut}, we will need to find the
Morse index of the $n$-th eigenpair $(\lambda_n,\psi)$ of $H^0$.
The lemma below will help us to do just that.

\begin{lemma}
  \label{lem:index}
  Let $A$ be a bounded from below self-adjoint operator acting on a
  real Hilbert space $\cH$.  Assume that $A$ has only discrete
  spectrum below a certain $\Lambda$ and its eigenvalues are ordered
  in increasing order.  Let $h[f]$ be the quadratic form
  corresponding to $A$.  If the $n$-th eigenvalue $\lambda_n <
  \Lambda$ is simple and $\psi$ is the corresponding eigenfunction,
  then the \emph{Lagrange functional}
  \begin{equation}
    \label{eq:Lagrange_functional}
    L(\lambda,f) = h[f] - \lambda( \|f\|^2 - 1 )
  \end{equation}
  has a non-degenerate critical point at $(\lambda_n, \psi)$ whose Morse index is $n$.
\end{lemma}

\begin{proof}
  We split the Hilbert space $\cH$ into the orthogonal sum $\cH_{-}
  \oplus \cH_0 \oplus \cH_{+}$.  Here the space $\cH_{-}$ is the span
  of the first $n-1$ eigenfunctions of $A$, the space $\cH_0$ is the
  span of the $n$-th eigenfunction $\psi$, and $\cH_{+}$ is their
  orthogonal complement.  The quadratic form $h$ is reduced by the
  decomposition $\cH = \cH_{-} \oplus \cH_0 \oplus \cH_{+}$, namely,
  \begin{equation*}
    h\left[f_{-} + f_0 + f_{+}\right]
    =  h[f_{-}] + h[f_0] + h[f_{+}].
  \end{equation*}
  On $\cH_{-}$, the quadratic form $h$ is bounded from above,
  \begin{equation*}
    h[f_{-}] \leq \lambda_{n-1} \|f_{-}\|^2 < \lambda_{n} \|f_{-}\|^2.
  \end{equation*}
  Similarly, on $\cH_{+}$ the form $h$ is bounded from below,
  \begin{equation*}
    h[f_{+}] > \lambda_{n} \|f_{+}\|^2.
  \end{equation*}
  Finally, on $\cH_0$ we have
  \begin{equation*}
    h[f_0] = \lambda_n s^2,
  \end{equation*}
  where $f_0 = s\psi$, $s\in\R$.

  To show that $(\lambda_n, \psi)$ is a critical point and calculate
  its index we evaluate
  \begin{equation*}
    \delta L := L(\lambda_n + \delta \lambda, \psi + \delta f) - L(\lambda_n, \psi)
    = L(\lambda_n + \delta \lambda, \psi + \delta f) - \lambda_n.
  \end{equation*}
  Expanding $\delta f = \delta f_{-} + s\psi + \delta f_{+}$ according
  to our decomposition of $\cH$, we see that
  \begin{align*}
    \delta L
    = h\left[\delta f_{-}\right] &+ \lambda_n (1+s)^2 +
    h\left[\delta f_{+}\right] \\
    &- (\lambda_n + \delta \lambda) \left( \|\delta f_{-}\|^2  + (1+s)^2
      + \|\delta f_{+}\|^2 - 1\right) - \lambda_n.
  \end{align*}
  Simplifying and completing squares, we obtain
  \begin{align*}
    \delta L
    = \left(h\left[\delta f_{-}\right] - \lambda_n \|\delta f_{-}\|^2 \right)
    &- \frac12 (s+\delta\lambda)^2 
    + \frac12 (s-\delta\lambda)^2 \\
    &+ \left(h\left[\delta f_{+}\right] - \lambda_n \|\delta f_{+}\|^2 \right)
    + \mbox{ higher order terms},
  \end{align*}
  where the two middle terms are representing $2s\delta\lambda$.  We
  observe that all terms are quadratic or higher order and hence
  $(\lambda_n,\psi)$ is a critical point as claimed.  The first two
  terms represent the negative part of the Hessian.  Their dimension
  is the dimension of $\cH_{-}$ plus one.  Thus the Morse index is
  $(n-1)+1 = n$.
\end{proof}

We remark that in the finite-dimensional case the Hessian of $\delta
L$ at the critical point is known as the ``bordered Hessian'' (see
\cite{Spr_amm85} for a brief history of the term).

\subsection{Restriction to a critical manifold}\label{sec:reduction}

We will also use the following simple result from \cite{BerKucSmi_gafa12}.

\begin{lemma}
  \label{T:reduction}
  Let $X=Y \bigoplus Y'$ be a direct decomposition of a Banach space.  Let also
  $f:X\to \R$ be a smooth functional such that $(0,0)\in X$ is its
  critical point of Morse index $\ind(f)$.

  If for any $y$ in a neighborhood of zero in $Y$, the point $(y,0)$
  is a critical point of $f$ over the affine subspace $\{y\}\times
  Y'$, then the Hessian of $f$ at the origin, as a quadratic form in
  $X$, is reduced by the decomposition $X=Y\bigoplus Y'$.

  In particular,
  \begin{equation}
    \label{eq:indices_add}
    \ind(f) = \ind(f|_{Y}) + \ind(f|_{Y'}),
  \end{equation}
  where $\ind(f|_W)$ is the Morse index of $0$ as the critical point
  of the function $f$ restricted to the space $W$.  Moreover, if
  $(0,0)$ is a non-degenerate critical point of $f$ on $X$, then $0$
  is non-degenerate as a critical point of $f|_{Y}$.
\end{lemma}

The subspace $Y$, which is the locus of the critical points of $f$
over the affine subspaces $(y,\cdot)$, is called the \emph{critical
  manifold}.  In applications, the locus of the critical points with
respect to a chosen direction is usually not a linear subspace.  Then
a simple change of variables is applied to reduce the situation to that
of Lemma~\ref{T:reduction}, while the Morse index remains unchanged.

\subsection{Proof of Theorem~\ref{thm:main_cut}}
\label{sec:proof_main_cut_proof}

In this subsection, $\lambda$ is used as both an independent variable
and a function (eigenvalue as a function of parameters).  To reduce
the confusion we denote $\xi = \lambda_n(H^0)$.  Recall that $\psi$ is the
$n$-th eigenfunction of $H^0$ and $\phi$ denotes the number
of internal zeros of $\psi$ on $\Gamma$.

\begin{proof}[Proof of Part~\ref{item:same_eig} of Theorem~\ref{thm:main_cut}]
  By design, $\psi$ is an eigenfunction of $H_{\boldsymbol{\gamma}}$ when
  $\boldsymbol{\gamma}=\widetilde{\boldsymbol\gamma}$; the vertex conditions at the new
  vertices $c_j^\pm$ were specifically chosen to fit $\psi$.  We
  conclude that $\xi \in \sigma\left(H_{\widetilde{\boldsymbol\gamma}}\right)$.

  Since $\psi$ is non-zero on vertices, the corresponding eigenvalue
  of $H_{\widetilde{\boldsymbol\gamma}}$ is simple \cite{Sch_wrcm06,BerKuc_incol12}.
  Eigenfunctions on a tree are Courant-sharp
  \cite{poketal_mat96,Sch_wrcm06,BerKuc_incol12}; in other words the
  eigenfunction number $n$ has $n-1$ internal zeros.  We use this
  property in reverse, concluding that $\psi$ is the eigenfunction
  number $\phi + 1$ of the tree operator $H_{\widetilde{\boldsymbol\gamma}}$.
\end{proof}

\begin{proof}[Proof of Part 2 of Theorem~\ref{thm:main_cut}]
  Here we prove that $\boldsymbol{\gamma} = \widetilde{\boldsymbol\gamma}$ is a
  critical point of $\lambda_{\phi+1}(H_{\boldsymbol{\gamma}})$.

  Consider the Lagrange functional
  \begin{equation}
    \label{F3}
    F_3(\lambda, f, \boldsymbol{\gamma}) = h_{\boldsymbol{\gamma}}[f]
    - \lambda\left(\sum_{e \in E(T)}\int_e f^2(x)\,dx - 1\right)
  \end{equation}
  where $h_{\boldsymbol{\gamma}}[f]$ is the quadratic form of the operator
  $H_{\boldsymbol{\gamma}}$, given by equation~\eqref{eq:qfhgamma}.  Observe
  that $\lambda_{\phi+1}(H_{\boldsymbol{\gamma}}) =: \lambda(\boldsymbol{\gamma})$ is
  a restriction of $F_3$ onto a
  submanifold, namely,
  \begin{equation}
    \label{eq:restriction_of_F3}
    \lambda_n(\boldsymbol{\gamma}) = F_3\Big(\lambda(\boldsymbol{\gamma}), f(\boldsymbol{\gamma}),
    \boldsymbol{\gamma}\Big),
  \end{equation}
  where $f(\boldsymbol{\gamma})$ is the normalized $(\phi+1)$-th eigenfunction of
  $H_{\boldsymbol{\gamma}}$.  We will now show that $(\xi, \psi,
  \widetilde{\boldsymbol\gamma})$ is a critical point of $F_3$; then criticality
  of the function $\lambda_{\phi+1}(H_{\boldsymbol{\gamma}})$ will follow
  immediately.

  We know from Lemma \ref{lem:index} that the the eigenpair $(\xi,
  \psi)$ is a critical point of the Lagrange functional and therefore
  \begin{equation*}
    \left.\frac{\partial F_3}{\partial \lambda}\right|_{(\lambda, f,
      \boldsymbol{\gamma}) = (\xi, \psi, \widetilde{\boldsymbol\gamma})} = 0
    \hspace{.3in}\mbox{and}\hspace{.3in}
    \left.\frac{\partial F_3}{\partial f}\right|_{(\lambda, f,
      \boldsymbol{\gamma}) = (\xi, \psi, \widetilde{\boldsymbol\gamma})} = 0.
  \end{equation*}
  Additionally, we calculate from equation~(\ref{eq:qfhgamma}) that
  \begin{equation*}
    \left. \frac{\partial F_3}{\partial \gamma_j}\right|_{(\lambda, f, \boldsymbol{\gamma}) =
      (\xi, \psi, \widetilde{\boldsymbol\gamma})} = \psi^2(c_j^+) - \psi^2(c_j^-)=0
    \hspace{.2in} \mbox{ for } j=1,\ldots,\beta,
  \end{equation*}
  since $\psi \in \widetilde{H}^1(\Gamma,\R)$ is continuous at all cut
  points $c_j$.  This proves that $\widetilde{\boldsymbol\gamma}$ is a critical point
  of $\lambda_{\phi+1}(H_{\boldsymbol{\gamma}})$. The non-degeneracy of the point will follow from the proof of part 3.
\end{proof}

\begin{proof}[Proof of Part 3 of Theorem~\ref{thm:main_cut}]
  We will calculate the index of the critical point
  $\widetilde{\boldsymbol\gamma}$ of $\lambda_{\phi+1}(H_{\boldsymbol{\gamma}})$ in two
  steps.  We will first establish that the index of $(\xi, \psi,
  \widetilde{\boldsymbol\gamma})$ as a critical point of $F_3$ is equal to
  $n+\beta$.  Then we will apply Lemma~\ref{T:reduction}
  to the restriction introduced in ~(\ref{eq:restriction_of_F3})
  in order to deduce the final result.  In fact the second step is simpler
  and we start with it to illustrate our technique.

  \textbf{Index of the critical point $\boldsymbol{\gamma} =
    \widetilde{\boldsymbol\gamma}$ of $\lambda(H_{\boldsymbol{\gamma}})$.}
  Assume we have already shown that $(\xi, \psi,
  \widetilde{\boldsymbol\gamma})$ is a non-degenerate critical point of $F_3$ of
  index $n+\beta$.
  Define the following change of variables:
  \begin{equation}
    \begin{cases}
      \hat{\lambda} = \lambda - \lambda(\boldsymbol{\gamma}),\\
      \hat{f} = f - f(\boldsymbol{\gamma}),\\
      \hat{\boldsymbol\gamma} = \boldsymbol{\gamma} - \widetilde{\boldsymbol\gamma},
    \end{cases}
  \end{equation}
  where $\lambda(\boldsymbol{\gamma})$ is the $(\phi+1)$-th eigenvalue
  of the operator $H_{\boldsymbol{\gamma}}$ and
  $f(\boldsymbol{\gamma})$ is the corresponding normalized
  eigenfunction.  The eigenvalue $\lambda(\boldsymbol{\gamma})$ is
  simple when $\boldsymbol{\gamma} = \widetilde{\boldsymbol\gamma}$
  (see the proof of Part 1 above) and this property is preserved
  locally.

  The critical point $(\xi, \psi, \widetilde{\boldsymbol\gamma})$ corresponds, in
  the new variables, to $(0,0,\boldsymbol{0})$.  The change of variables is
  obviously non-degenerate and therefore the signature of a critical
  point remains unchanged.

  For every fixed $\boldsymbol{\gamma}$, the function $F_3$ is the
  Lagrange functional of the operator $H_{\boldsymbol{\gamma}}$ and by
  Lemma~\ref{lem:index} we conclude that
  $(\lambda(\boldsymbol{\gamma}), f(\boldsymbol{\gamma}))$ is its
  non-degenerate critical point of index $\phi+1$.  In the new
  variables this translates to $(0, 0, \hat{\boldsymbol\gamma})$ being
  a critical point with respect to the first two variables for any
  value of the third variable.  Now we can apply
  Lemma~\ref{T:reduction} to conclude that $\hat{\boldsymbol\gamma} =
  \boldsymbol{0}$ is a non-degenerate critical point of $F_3(0, 0,
  \hat{\boldsymbol\gamma})$ with index $(n+\beta) - (\phi+1)$.

  Since $F_3(0, 0, \hat{\boldsymbol\gamma}) = \lambda_{\phi+1}(\boldsymbol{\gamma})$, we
  obtain the desired conclusion.  It remains to verify the assumption
  that $(\xi, \psi, \widetilde{\boldsymbol\gamma})$ is a non-degenerate critical
  point of $F_3$ of index $n+\beta$.

  \textbf{Index of critical point of $F_3$.} By Remark
  \ref{rem:functiondecomp}, any $f \in D_{\boldsymbol\gamma}$ can be written as
  \begin{equation*}
    f = f_0 + \sum_{j=1}^\beta s_j\rho_{j,\lambda}
    =: f_0 + \boldsymbol{s}\cdot\boldsymbol{\rho_{{}_\lambda}},
  \end{equation*}
  where $f_0 \in D$ and each $\rho_{j,\lambda}$ satisfies $g''(x) +
  (q(x) - \lambda)g(x) = 0$.  Therefore the Lagrange functional $F_3$
  can be re-parametrized as follows:
  \begin{align*}\label{eq:F4}
    F_4(\lambda, f_0, \boldsymbol{s}, \boldsymbol{\gamma}) &:=
    F_3\left(\lambda,  f_0 + \boldsymbol{s}\cdot\boldsymbol{\rho_{_\lambda}}, \boldsymbol{\gamma}\right)\\
    &= h_{\boldsymbol{\gamma}}[f_0 + \boldsymbol{s}\cdot\boldsymbol{\rho_{{}_\lambda}}]
    - \lambda\left(\int_T \left(f_0
        + \boldsymbol{s}\cdot\boldsymbol{\rho_{_\lambda}}\right)^2\, dx - 1\right),
  \end{align*}
  where we understand the integral over the graph $T$ as the sum of
  integrals over all edges of $T$.  We let
  \begin{equation*}
    R_j(\lambda) = \rho_{j,\lambda}'(c_j^+) + \rho_{j,\lambda}'(c_j^-),
  \end{equation*}
  and investigate $F_4$ as a function of $\boldsymbol{s}$ and $\boldsymbol{\gamma}$ while
  $\lambda$ and $f_0$ are held fixed.  It turns out that $(\boldsymbol{s},
  \boldsymbol{\gamma}) = (\boldsymbol{0}, \boldsymbol{R})$ is a critical point.  Indeed, we
  calculate explicitly that
  \begin{align*}
    \left. \frac12 \frac{\partial F_4}{\partial s_j}\right|_{\boldsymbol{s} = 0}
    &=  \int_{T} \left(f_0'\rho_{j,\lambda}'
      + qf_0\rho_{j,\lambda}
      - \lambda  f_0\rho_{j,\lambda}\right)\, dx \\
    &\hspace{3cm}+ \gamma_j f_0(c_j^+)(\rho_{j,\lambda}(c_j^+)-\rho_{j,\lambda}(c_j^-))\\
    &=  \int_{T} f_0\left(-\rho_{j,\lambda}''
        + (q-\lambda)\rho_{j,\lambda}\right)\, dx \\
    &\hspace{3cm}  + f_0(x)\rho_{j,\lambda}'(x)\Big|_{c_j^-}^{c_j^+} + \gamma_j f_0(c_j^+)\\
    &=f_0(c_j^+)\left(\gamma_j -\rho_{j,\lambda}'(c_j^+) - \rho_{j,\lambda}'(c_j^-) \right)
  \end{align*}
  is equal to zero when $\gamma_j = R_j(\lambda)$.  Note that we
  assumed $\chi_v=0$ at every vertex of the graph $\Gamma$.  If
  this is not the case, the corresponding terms cancel out when the
  integration by parts is performed.

  The partial derivatives with respect to $\gamma_j$ also vanish,
  \begin{equation*}
    \left.\frac{\partial F_4}{\partial \gamma_j}\right|_{\boldsymbol{s} = 0}
    = \left[f^2(c_j^+) - f^2(c_j^-) \right]_{\boldsymbol{s}=0}
    = f_0^2(c_j^+) - f_0^2(c_j^-) = 0,
  \end{equation*}
  since the function $f_0 \in D$ is continuous across the cut vertices
  $c_j$.  Here we used the short-hand $f = f_0 +
  \boldsymbol{s}\cdot\boldsymbol{\rho_{_\lambda}}$.

  We can also calculate the Morse index of the critical point $(\boldsymbol{s},
  \boldsymbol{\gamma}) = (\boldsymbol{0}, \boldsymbol{R})$.  The Hessian is block-diagonal with
  $\beta$ blocks of the form
  \begin{equation*}
    \begin{bmatrix}
      \frac{\partial^2 F}{\partial \gamma_j \partial \gamma_j}
      & \frac{\partial^2 F}{\partial s_j \partial \gamma_j}\\
      \frac{\partial^2 F}{\partial \gamma_j \partial s_j}
      & \frac{\partial^2 F}{\partial s_j \partial s_j}\\
    \end{bmatrix}
    =
    \begin{bmatrix}
      0 & 2f_0(c_j^+)\\
      2f_0(c_j^+) & \cdot 
    \end{bmatrix},
  \end{equation*}
  where the value of the second derivative with respect to $s_j$ is
  irrelevant.  Each block has negative determinant and therefore
  contributes one negative and one positive eigenvalue.  The total
  index is therefore $\beta$ and the critical point is obviously non-degenerate.

  Finally, we observe that the critical manifold $(\lambda, f_0, \boldsymbol{0},
  \boldsymbol{R}(\lambda))$ passes through the critical point $(\xi, \psi, \boldsymbol{0},
  \widetilde{\boldsymbol\gamma})$.  To show this we need to verify that $R_j(\xi)
  = \widetilde{\boldsymbol\gamma}_j$ when $f_0=\psi$.  Applying
  Lemma~\ref{lem:wronskian_cutpoints} to the Wronskian of $\psi$ and
  $\rho_j$ we obtain
  \begin{equation*}
    \psi'(c_j^+)\rho_{j,\xi}(c_j^+) - \rho'_{j,\xi}(c_j^+)\psi(c_j^+)
    = -\psi'(c_j^-)\rho_{j,\lambda}(c_j^-) + \rho'_{j,\xi}(c_j^-)\psi(c_j^-).
  \end{equation*}
  Substituting the boundary values of $\rho_j$ (see
  Remark~\ref{rem:functiondecomp}), we arrive at
  \begin{equation*}
    \tilde{\gamma_j} := \frac{\psi'(c_j^+)}{\psi(c_j^+)}
    = \rho'_{j,\xi}(c_j^+) + \rho'_{j,\xi}(c_j^-) =: R_j.
  \end{equation*}

  By using the non-degenerate change of variables
  \begin{equation}
    \begin{cases}
      \hat{\lambda} = \lambda - \xi,\\
      \hat{f} = f_0 - \psi,\\
      \hat{\boldsymbol{s}} = \boldsymbol{s} - \boldsymbol{0},\\
      \hat{\boldsymbol\gamma} = \boldsymbol{\gamma} - \boldsymbol{R}(\lambda),
    \end{cases}
  \end{equation}
  we can again apply Lemma~\ref{T:reduction} with $Y =
  \big\{(\hat{\lambda}, \hat{f}, \boldsymbol{0}, \boldsymbol{0}) \big\}$.  On the subspace $Y$,
  the function $F_3$ is equal to $h[f_0] - \lambda
  \left(\|f_0\|^2-1\right)$, which is precisely the Lagrange
  functional for the operator $H^0$ with the correct domain.  By
  Lemma~\ref{lem:index} it has index $n$ at the point
  $(\lambda,f_0) =(\xi,\psi)$.  Adding the two indices together we obtain
  index $n+\beta$ for the critical point $(\xi, \psi, \boldsymbol{0},
  \widetilde{\boldsymbol\gamma})$ of $F_4$, which corresponds to the critical
  point $(\xi, \psi, \widetilde{\boldsymbol\gamma})$ of $F_3$.  This concludes
  our proof.
\end{proof}


\section{Critical Points of $\lambda_n(H^{\boldsymbol\alpha})$}

In this section we show that $\boldsymbol{\alpha} = (0, \ldots, 0)$ is a
critical point of $\lambda(H^{\boldsymbol{\alpha}})$ and compute its Morse
index, thus concluding the proof of Theorem~\ref{thm:main_mag}.


\subsection{Points of symmetry}

\begin{theorem}\label{neglambda}
  Let $\sigma(\boldsymbol{\alpha})$ denote the spectrum of $H^{\boldsymbol{\alpha}}$ where $\boldsymbol{\alpha} = (\alpha_1,
  \ldots, \alpha_\beta) \in \R^\beta$.  Then all points in the set
  \begin{equation}
    \label{eq:symmetry_points}
    \Sigma = \left\{ \pi (b_1, \ldots, b_\beta) : b_j \in \{0,1\} \right\}
  \end{equation}
  are points of symmetry of $\sigma(\boldsymbol{\alpha})$, i.e.\ for all $\boldsymbol{\alpha} \in \R^\beta$
   and for all $\boldsymbol{\varsigma} \in \Sigma$,
  \begin{equation}
    \label{eq:symmetry_points_def}
    \sigma(\boldsymbol{\varsigma}-\boldsymbol{\alpha}) = \sigma(\boldsymbol{\varsigma} + \boldsymbol{\alpha}),
  \end{equation}
  together with multiplicity.

  Consequently, if $\lambda_n(\boldsymbol{\alpha})$ is the $n$-th eigenvalue
  of $H^{\boldsymbol{\alpha}}$ that is simple at $\boldsymbol{\alpha} = \boldsymbol{\varsigma} \in
  \Sigma$, then $\boldsymbol{\varsigma}$ is a critical point of the function
  $\lambda_n(\boldsymbol{\alpha})$.
\end{theorem}

\begin{proof}
  We will show that if $f(x)$ is an eigenfunction of
  $H^{\boldsymbol{\varsigma} - \boldsymbol{\alpha}}$, then $\overline{f(x)}$ is an
  eigenfunction of $H^{\boldsymbol{\varsigma}+\boldsymbol{\alpha}}$.  Since the
  operator $H^{\boldsymbol{\alpha}}$ is self-adjoint, we know that the
  eigenvalues are real.  Taking the complex conjugate of the
  eigenvalue equation for $f$ we see that $\overline{f(x)}$ satisfies
  the same equation,
  \begin{equation*}
    -\frac{d^2\overline{f(x)}}{dx^2} + q(x)\overline{f(x)} =  \lambda\overline{f(x)}.
  \end{equation*}
  Similarly, all vertex conditions at the vertices of
  the tree $T$ inherited from $\Gamma$ have real coefficients and therefore $\overline{f(x)}$
  satisfies them too.  The only change occurs at the vertices
  $c^\pm_j$.

  Note that for every $\boldsymbol{\sigma} \in \Sigma$, $\sigma_j$ is equal to either $0$ or $\pi$ so
  $e^{2i\sigma_j} = 1$ for all $j=1,\ldots,\beta$.  Therefore,
  \begin{equation*}
     \overline{e^{i(\sigma_j - \alpha_j)}} = e^{i(-\sigma_j + \alpha_j)}
     = e^{i(\sigma_j + \alpha_j)} e^{-2i\sigma_j}
     = e^{i(\sigma_j + \alpha_j)}.
  \end{equation*}
  Conjugating the vertex conditions of
  $H^{\boldsymbol{\varsigma}-\boldsymbol\alpha}$ at $c^\pm_j$ we obtain
  \begin{equation*}
    \overline{f(c_j^-)}
    = \overline{e^{i(\sigma_j - \alpha_j)}f(c_j^+)}
    = e^{i(\sigma_j + \alpha_j)}\overline{f(c_j^+)},
  \end{equation*}
  and same for the derivative.  Thus $\overline{f(x)}$ satisfies the
  vertex conditions of the operator $H^{\boldsymbol{\varsigma}+\boldsymbol\alpha}$
  and vice versa.  The spectra of these two operators are therefore
  identical.
\end{proof}

\subsection{A non-self-adjoint continuation}
We now consider the same operator $-\frac{d^2}{dx^2} + q(x)$ on the
tree $T$ with different vertex conditions at $c_j^\pm$:
\begin{equation}
  \label{eq:real_jump_cond}
  \begin{split}
    f(c_j^+) &= e^{\alpha_j} f(c_j^-), \\
    f'(c_j^+) &= -e^{\alpha_j} f'(c_j^-),
  \end{split}
\end{equation}
i.e.\ the function has a jump in magnitude across the cut.  It is easy
to see that these conditions are obtained from
\eqref{eq:phase_jump_cond} by changing $\boldsymbol{\alpha}$ to
$-i\boldsymbol{\alpha}$.  We will denote the operator with vertex conditions
$\eqref{eq:real_jump_cond}$ at the vertices in $T\backslash\Gamma$ by
$H^{i\boldsymbol{\alpha}}$.

\begin{remark}
  The operator of $H^{i\boldsymbol{\alpha}}$ is not self-adjoint for $\boldsymbol{\alpha} \in \R^\beta$.
  A simple example is the interval $[0,\pi]$ with $q(x) = 0$ and conditions
  \begin{equation*}
    f(0) = e^\alpha f(\pi) \qquad \mbox{and}
    \qquad f'(0) = e^\alpha f'(\pi),
  \end{equation*}
  which has complex eigenvalues when $\alpha \neq 0$.  Indeed,
  the eigenvalues are easily calculated to be
  \begin{equation*}
    \lambda_n = \left(2n \pm \frac{\alpha}{\pi}i\right)^2.
  \end{equation*}
\end{remark}

\begin{lemma}
  If $\lambda_n(H^0)$ is simple, then locally around $\boldsymbol{\alpha} =
  (0,\ldots, 0)$ the eigenvalue $\lambda_n(H^{i\boldsymbol{\alpha}})$ is
  real.  The corresponding eigenfunction is real too.
\end{lemma}

\begin{proof}
  By standard perturbation theory \cite{Kato_book} (see also
  \cite{BerKuc_incol12} for results specifically on graphs) we know
  that $\lambda_n(H^{i\boldsymbol{\alpha}})$ is an analytic function of
  $\boldsymbol\alpha$ and since $\lambda_n(H^0)$ is simple,
  $\lambda_n(H^{i\boldsymbol{\alpha}})$ remains simple in a neighborhood of
  $\boldsymbol{\alpha} = (0,\ldots, 0)$.  Since the operator
  $H^{i\boldsymbol{\alpha}}$ has real coefficients, its complex eigenvalues
  must come in conjugate pairs.  For this to happen, the real
  eigenvalue must first become double.  Since
  $\lambda_n(H^{i\boldsymbol{\alpha}})$ is simple near $(0,\ldots,0)$, the
  eigenvalue is real there.
\end{proof}

We note that since we impose no restrictions on the eigenvalues below
$\lambda_n(H^0)$, some of them might turn complex as soon as
$\boldsymbol\alpha\neq 0$.  In this case, the ``$n$-th'' eigenvalue
$\lambda_n(H^{i\boldsymbol{\alpha}})$ refers to the unique continuation
of $\lambda_n(H^0)$.  Locally, of course, it is the same as having the
eigenvalues ordered by their real part.

\subsection{Connection between $H_{\boldsymbol\gamma}$ and $H^{i\boldsymbol\alpha}$}
\label{sec:mappingR}

Locally around $\boldsymbol{\gamma}=\widetilde{\boldsymbol\gamma}$ we introduce a
mapping $R: \boldsymbol{\gamma} \mapsto \boldsymbol{\alpha}$ so that
$\lambda_{\phi+1}\left(H_{\boldsymbol{\gamma}}\right)
= \lambda_n\left(H^{i\boldsymbol{\alpha}}\right)$
when $R(\boldsymbol{\gamma}) = \boldsymbol{\alpha}$.

For a given $\boldsymbol{\gamma}$, we find the $(\phi+1)$-th eigenfunction of
$H_{\boldsymbol{\gamma}}$, denoting it by $g$.  We then let
\begin{equation}
  \label{eq:gamma_to_alpha}
  e^{\alpha_j} = \frac{g(c_j^+)}{g(c_j^-)}
\end{equation}
and, since $g$ satisfies equation~\eqref{eq:cut_cond}, it is now
easy to check that $g$ is indeed an eigenfunction of
$H^{i\boldsymbol{\alpha}}$.

\begin{lemma}
  The function $R$ is a non-degenerate diffeomorphism.  Therefore, the
  point $\boldsymbol{\alpha} = (0,\ldots, 0)$ is a critical point of the
  function $\lambda_n\left(H^{i\boldsymbol{\alpha}}\right)$ of index
  $n-1+\beta-\phi$.
\end{lemma}

\begin{proof}
  The function $R$ is an analytic function in a neighborhood of
  $\widetilde{\boldsymbol\gamma}$ since the eigenfunctions are analytic functions of
  the parameters and therefore $R$ is a composition of analytic
  functions.  We can define $R^{-1}$ by reversing the process,
  i.e.\ for a given $\boldsymbol\alpha$ find the (real) $n$-th eigenfunction
  $\varphi$ of $H^{i\boldsymbol{\alpha}}$ and let $\gamma_j = \varphi'(c^+_j)
  / \varphi(c^+_j)$.  By the same arguments, $R^{-1}$ is also an
  analytic function in a neighborhood of $(0,\ldots, 0)$.  Therefore
  $R$ is a non-degenerate diffeomorphism.

  A diffeomorphism preserves the index and therefore the index of
  $(0,\ldots, 0)$ of the function
  $\lambda_n\left(H^{i\boldsymbol{\alpha}}\right)$ is the same as the index of
  $\widetilde{\boldsymbol\gamma}$ of the function
  $\lambda_{\phi+1}\left(H_{\boldsymbol{\gamma}}\right)$, which was computed
  in Theorem~\ref{thm:main_cut}.
\end{proof}

\subsection{From $H^{i\boldsymbol\alpha}$ to $H^{\boldsymbol\alpha}$}

\begin{proof}[Proof of Theorem~\ref{thm:main_mag}]
  The function $\lambda_n(H^{\boldsymbol{\alpha}})-\xi$ is analytic and,
  locally around $(0,\ldots, 0)$, quadratic in $\{\alpha_j\}$ because
  $(0, \ldots, 0)$ is a critical point so the linear term (the first
  derivative) is zero.  Substituting $\boldsymbol{\alpha} \to
  -i\boldsymbol{\alpha}$ into the quadratic term results in an overall minus,
  that is
  \begin{equation*}
    \lambda_n(H^{i\boldsymbol{\alpha}}) - \xi
    = -\left(\lambda_n(H^{\boldsymbol{\alpha}}) - \xi\right) +
    \mbox{higher order terms}.
  \end{equation*}
  Therefore the index of $(0,\ldots, 0)$ as a critical point of
  $\lambda_n(H^{\boldsymbol{\alpha}})$ is the dimension of the space of
  variables minus the index of $\lambda_n(H^{i\boldsymbol{\alpha}})$.
   Thus it is equal to
  \begin{equation*}
    \beta - (n-1+\beta-\phi) = \phi - (n-1).
  \end{equation*}
\end{proof}


\section{Connection to partitions on graphs}
\label{sec:partition}

The set of points on which a real eigenfunction vanishes (called the
\emph{nodal set}) generically has co-dimension 1.  Thus, when one
considers a problem which is not 1-dimensional (or
quasi-1-dimensional, like a graph), counting the \emph{number} of
zeros does not make sense. Then one usually counts the number of
``nodal domains'': the connected components obtained after removing
the nodal set from the domain.  We refer to the number of nodal
domains as the \emph{nodal domain count}.  It should be noted that the
nodal domain count is a non-local property \cite{BanOreSmi_pspm08}.
Let $\nu_n$ denote the number of nodal domains of the $n$-th
eigenfunction. Then a classical result of Courant
\cite{Cou_ngwgmp23,CourantHilbert_volume1}, in the case of the Dirichlet
Laplacian, bounds $\nu_n$ from above by $n$, independently of
dimension.

An interesting new point of view on the nodal domains arose recently,
see \cite{HelHofTer_aihp09} and references therein.  Namely, a domain
is partitioned into subdomains and the following question is asked:
when does a given partition coincide with the nodal partition
corresponding to an eigenfunction of the Dirichlet Laplacian on the
original domain?  It turns out that there is a natural ``energy''
functional defined on partitions whose minima correspond to the
eigenfunctions satisfying $\nu_n=n$.  Restricting the set of allowed
partitions, it was further found
\cite{Banetal_cmp12,BerRazSmi_jpa12,BerKucSmi_gafa12} that all
critical points of this functional correspond to eigenfunctions and
the ``nodal deficiency'' $n-\nu_n$ is equal to the Morse index of the
critical point (which is zero for a minimum).

The latter result was first established on graphs in
\cite{Banetal_cmp12} and here we outline how its strengthened version
follows from our Theorem~\ref{thm:main_cut}.  We define a \emph{proper
  $m$-partition} $P$ of a graph $\Gamma$ as a set of $m$ points lying
on the edges of the graph (and not on the vertices).  Enforcing
Dirichlet conditions at these vertices effectively separates the graph
$\Gamma$ into \emph{partition subgraphs} which we will denote
$\Gamma_j$.  The functional mentioned above is defined as
\begin{equation}
  \label{eq:Lambda_def}
  \Lambda(P) := \max_{j} \lambda_1(\Gamma_j),
\end{equation}
where $\lambda_1(\Gamma_j)$ is the first eigenvalue of the partition
subgraph $\Gamma_j$.  The conditions on the vertices of $\Gamma_j$ are
either inherited from $\Gamma$ or taken to be Dirichlet on the newly
formed vertices.

The partition $P$ should be understood as a candidate for the nodal
set of an eigenfunction of $\Gamma$.  It is easy to see that the
partition points break every cycle of $\Gamma$ if and only if the
number $\nu(P)$ of the partition subgraphs is related to $m$ by
\begin{equation}
  \label{eq:m_and_nu}
  \nu(P) = m - \beta + 1.
\end{equation}
We start by considering the partitions and eigenfunctions that satisfy
the above property.  In section~\ref{sec:few_zeros} we will treat the
case of partitions where some of the cycles survive.

Further, we call an $m$-partition an \emph{equipartition} if all
subgraphs $\Gamma_j$  share the same eigenvalue:
\begin{equation*}
  \lambda_1(\Gamma_{j_1}) = \lambda_1(\Gamma_{j_2}),
  \qquad \mbox{for all }j_1, j_2.
\end{equation*}
It is easy to see that the partition defined by the nodal set of an
eigenfunction is an equipartition.
In \cite{Banetal_cmp12} it was shown that the set of
$m$-equipartitions on $\Gamma$ can be parametrized using $\beta$
parameters $\{\gamma_j\}$ and the operator $H_{\boldsymbol{\gamma}}$ defined
in section~\ref{sec:cutting}: we take the $(m+1)$-th eigenfunction of
$H_{\boldsymbol{\gamma}}$ and its zeros (transplanted to the original graph
$\Gamma$) define an equipartition.  With such parametrization, the
energy $\Lambda(P)$ of the partition is simply the $(m+1)$-th eigenvalue
$\lambda_{m+1}(H_{\boldsymbol\gamma})$.  Now Theorem~\ref{thm:main_cut}
immediately implies the following.

\begin{corollary}
  \label{cor:partition}
  Suppose the $n$-th eigenvalue of $\Gamma$ is simple and its
  eigenfunction $\psi$ is non-zero on vertices.  Denote by $\phi$ the
  number of zeros of $\psi$ and by $\nu$ the number of its nodal
  domains.  If the zeros of the eigenfunction break every cycle of
  $\Gamma$, then the $\phi$-partition defined by the zeros of $\psi$ is
  a non-degenerate critical point of the functional $\Lambda$ on the
  set of equipartitions.  The Morse index of the critical point is
  equal to $n-\nu$.
\end{corollary}

Some remarks are in order.  The ``converse'' fact that critical points
of $\Lambda$ correspond to eigenfunctions is easy to establish.  The
main difficulty lies in calculating the Morse index.  In the main
theorem of \cite{Banetal_cmp12}, the non-degeneracy of the critical
point had to be assumed \emph{a priori}.  In
section~\ref{sec:proof_main_cut}, we established that this actually
follows from the other assumptions.  Eigenfunctions
whose zeros do not break all cycles of $\Gamma$ correspond to low values
of $\lambda$, and it can easily be shown that there are only finitely many
such eigenfunctions.  We will handle these eigenfunctions by
introducing cut points only on those cycles which are broken by the
zeros of $\psi$ and correspondingly adjusting the operator
$H_{\boldsymbol\gamma}$.  Finally, the mapping $R$ defined in
section~\ref{sec:mappingR} essentially shows that the equipartitions
can be parameterized using eigenfunctions of the ``magnetic''
Schr\"odinger operator with purely imaginary magnetic field.

\subsection{Partitions with few zeros}
\label{sec:few_zeros}

For eigenfunctions corresponding to low eigenvalues, the nodal set
might not break all the cycles of the graph, see
Fig.~\ref{fig:cut_somecycles}(a).  In this case, the parameterization
of the nearby equipartitions is done via a modification of the
operator $H_\gamma$.  In this section we describe this
parameterization and point out the changes in the proofs of the
analogue of Theorem~\ref{thm:main_cut} that the new parameterization
necessitates.  An outline of the procedure has already appeared in
\cite{Banetal_cmp12,BerRazSmi_jpa12}; however some essential details
have been omitted there.

\begin{figure}[t]
  \centering
  \includegraphics[scale=0.5]{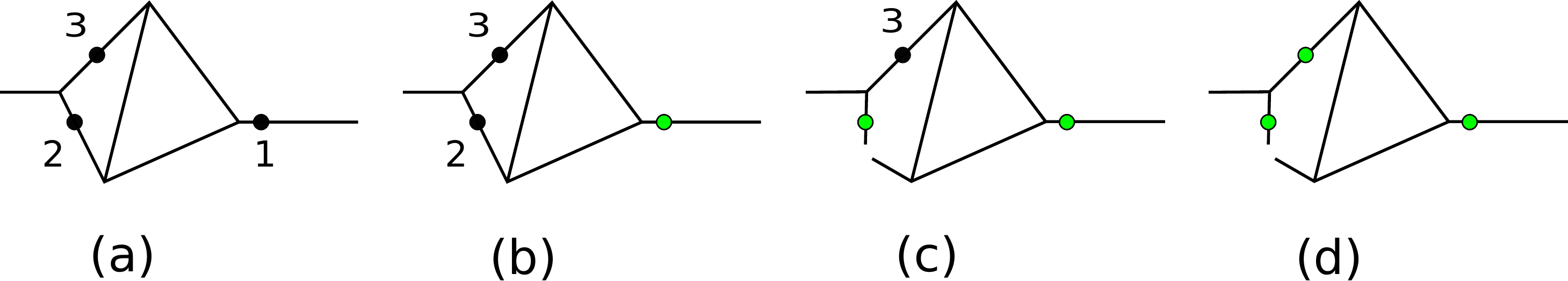}
  \caption{A partition with surviving cycles.  (a) Zeros, marked 1, 2,
    and 3, do not lie on all the cycles of the graph.  To find cut
    points we consider the zeros in sequence. (b) Cutting the graph at
    zero 1 would disconnect it.  (c) Cutting the graph at zero 2 would
    not disconnect it; therefore, a cut point is placed nearby. (d)
    Now, cutting the graph at zero 3 would disconnect the graph, so we
    do not introduce any more cut points.}
  \label{fig:cut_somecycles}
\end{figure}

As mentioned previously, the eigenfunctions we are interested in here
do not have a zero on every cycle.  Hence, unlike the previous case
for large eigenvalues where the corresponding eigenfunctions have at
least one zero on every cycle, we must carefully pick our cut points
to avoid cutting cycles that do not contain any zeros of the
eigenfunction.  To do this we look at the zeros of our eigenfunction
$\psi$ one at a time.  If cutting the edge that contains the zero will
disconnect the graph, we do nothing and remove this zero from
consideration (see Fig.~\ref{fig:cut_somecycles}(b)).  If cutting the
edge at the zero will not disconnect the graph, then we cut that edge
at a nearby point $c_j$ at which $\psi$ is non-zero, calling the new
vertices $c_j^+$ and $c_j^-$ as before (see
Fig.~\ref{fig:cut_somecycles}(c)).  Notice that the manner in which we
order and analyze the zeros does not matter; while the cut positions
and resulting graph may vary, we will make the same number of cuts.

Let us consider the number of cuts $\eta$ more explicitly.  Denote by
$\mathcal{N}$ the zero set of $\psi$ and remove $\mathcal{N}$ from
$\Gamma$ to get the (disconnected) graph
$\Gamma\backslash\mathcal{N}$.  Let $\nu$ be the number of connected
components \{$\Gamma_j$\} after the cutting (the components $\Gamma_j$
are the nodal domains of $\Gamma$ with respect to $\psi$).  Denote
\begin{equation*}
  \beta_{\Gamma\backslash\mathcal{N}} = \sum_{i=1}^\nu \beta_{\Gamma_{j}}
\end{equation*}
where $\beta_\kappa$ is the Betti number of the graph $\kappa$.  It is
easy to see that
\begin{equation}
  \label{eq:eta_alt_def}
  \eta = \beta_\Gamma - \beta_{\Gamma\backslash\mathcal{N}}
\end{equation}
and furthermore
\begin{equation}\label{eq:eta}
  \eta = 1 + \phi - \nu.
\end{equation}
For further details, see Lemma 5.2.1 of \cite{BK_graphs}.  Now we
continue with an alternative statement of Theorem \ref{thm:main_cut}.

\begin{theorem}\label{thm:main_cut_2}
  Let $\psi$ be the eigenfunction of $H^0$ that corresponds to a
  simple eigenvalue $\lambda_n(H^0)$.  We assume that $\psi$ is
  non-zero on internal vertices of the graph.  We denote by $\phi$ the
  number of internal zeros and $\nu$ the number of nodal domains of
  $\psi$ on $\Gamma$.  Let $c_j^\pm$, $j=1,\ldots,\eta$, be the cut
  points created by following the procedure above, where $\eta = 1 +
  \phi - \nu$.

  Let $H_{\boldsymbol{\gamma}}$, $\boldsymbol{\gamma} =
  (\gamma_1,\ldots,\gamma_\eta)$, be the operator obtained from $H^0$
  by imposing the additional conditions
  \begin{equation}  \label{eq:some_cut_cond}
    \begin{split}
      f'(c_j^+) &= \gamma_j f(c_j^+),\\
      f'(c_j^-) &= -\gamma_j f(c_j^-),
    \end{split}
  \end{equation}
  at the cut points.

  Define
  \begin{equation*}
    \tilde{\gamma}_j := \frac{\psi'(c_j^+)}{\psi(c_j^+)} = -\frac{\psi'(c_j^-)}{\psi(c_j^-)}
  \end{equation*}
  and let $\widetilde{\boldsymbol\gamma} =
  (\tilde{\gamma}_1,\ldots,\tilde{\gamma}_\eta)$.  Consider the
  eigenvalues of $H_{\boldsymbol{\gamma}}$ as functions
  $\lambda_n(H_{\boldsymbol{\gamma}})$ of $\boldsymbol{\gamma}$.  Then
  \begin{enumerate}
  \item
    $\lambda_{\phi+1}(H_{\boldsymbol{\gamma}})\Big|_{\boldsymbol{\gamma}
      = \widetilde{\boldsymbol\gamma}}
    = \lambda_n(H^0)$ where $\phi$ is the number of zeros of $\psi$
    on $\Gamma$,
  \item $\boldsymbol{\gamma} = \widetilde{\boldsymbol\gamma}$ is a non-degenerate
    critical point of the function
    $\lambda_{\phi+1}(H_{\boldsymbol{\gamma}})$, and
  \item the Morse index of the critical point
    $\boldsymbol{\gamma} = \widetilde{\boldsymbol\gamma}$
    is equal to $n-\nu$.
  \end{enumerate}
\end{theorem}

We will map out the proof of the theorem in
section~\ref{sec:few_zeros_proof} below, after explaining its
significance to the question of equipartitions.

\begin{theorem}
  \label{thm:partitions_gen}
  Suppose the $n$-th eigenvalue of $\Gamma$ is simple and its
  eigenfunction $\psi$ is non-zero on vertices.  Denote by $\phi$ the
  number of internal zeros of $\psi$ and by $\nu$ the number of its nodal
  domains.  Then the $\phi$-equipartitions in the vicinity of the
  nodal partition of $\psi$ are parametrized by the variables
  $\boldsymbol{\gamma} = (\gamma_1,\ldots,\gamma_\eta)$.

  The nodal partition of $\psi$ corresponds to the point
  $\widetilde{\boldsymbol\gamma} =
  (\tilde{\gamma}_1,\ldots,\tilde{\gamma}_\eta)$ and is a
  non-degenerate critical point of the functional $\Lambda$
  (equation~(\ref{eq:Lambda_def})) on the set of equipartitions.  The
  Morse index of the critical point is equal to $n-\nu$.
\end{theorem}

The mapping from $(\gamma_1,\ldots,\gamma_\eta)$ to the equipartitions
is constructed as follows (see \cite{Banetal_cmp12} for more details):
the partition in question is generated by the zeros of the
$(\phi+1)$-th eigenfunction of the operator
$H_{\boldsymbol{\gamma}}$ placed upon the original graph $\Gamma$.
Indeed, the groundstates of the nodal domains can be obtained by
cutting the eigenfunction at zeros and gluing the cut points together
(conditions \eqref{eq:some_cut_cond} ensure the gluing is possible).
To verify that all equipartitions are obtainable in this way we simply
reverse the process and construct an eigenfunction of
$H_{\boldsymbol{\gamma}}$ from the groundstates of the nodal domains.
The gluing is now done at zeros, and it can be done recursively (since
all cycles with zeros on them have been cut).

Once the parameterization of the equipartitions is accomplished, the
Morse index result follows immediately from
Theorem~\ref{thm:main_cut_2}.

\subsection{Proof of Theorem~\ref{thm:main_cut_2}}
\label{sec:few_zeros_proof}

The proof of Theorem~\ref{thm:main_cut_2} is identical to the proof of
Theorem~\ref{thm:main_cut} once we collect some preliminary results.
The following lemma can be found in \cite{BK_graphs} (Theorem 3.1.8
with a slight modification).

\begin{lemma}
  \label{lem:inequality}
  Let $\Gamma_{\alpha'}$ be the graph obtained from the graph
  $\Gamma_\alpha$ by changing the coefficient of the $\delta$-type
  condition at a vertex $v$ from $\alpha$ to $\alpha'$ (conditions at
  all other vertices are fixed).  If $-\infty < \alpha < \alpha' \leq
  \infty$ (where $\alpha' = \infty$ corresponds to the Dirichlet
  condition at vertex $v$), then
  \begin{equation}
    \label{eq:interlacing1}
    \lambda_n(\Gamma_\alpha) \leq \lambda_n(\Gamma_{\alpha'})
    \leq \lambda_{n+1}(\Gamma_\alpha).
  \end{equation}
  If the eigenvalue $\lambda_n(\Gamma_{\alpha'})$ is simple and its
  eigenfunction $f$ is such that either $f(v)$ or $\sum f'(v)$ is
  non-zero, then the above inequalities can be made strict.  If, in
  addition, $\alpha' \neq \infty$, the inequalities become
  \begin{equation*}
    \lambda_n(\Gamma_\alpha) < \lambda_n(\Gamma_{\alpha'})
    < \lambda_n(\Gamma_{\infty}) \leq \lambda_{n+1}(\Gamma_\alpha).
  \end{equation*}
\end{lemma}


The following theorem is a generalization of Corollary 3.1.9 of \cite{BK_graphs}.

\begin{theorem}
  \label{thm:simple}
  Let $\Gamma$ be a graph with $\delta$-type conditions at every
  internal vertex and extended $\delta$-type conditions on all leaves.
  Suppose an eigenvalue $\lambda$ of $\Gamma$ has an eigenfunction $f$
  which is non-zero on internal vertices of $\Gamma$.  Further, assume
  that no zeros of $f$ lie on the cycles of $\Gamma$.  Then the
  eigenvalue $\lambda$ is simple and $f$ is eigenfunction number $\phi
  + 1$, where $\phi$ is the number of internal zeros of $f$.
\end{theorem}

\begin{remark}
  \label{rem:no_zero_cycle}
  The condition that no zeros lie on the cycles of the graph $\Gamma$
  is equivalent to $\eta=0$ (see equation~(\ref{eq:eta_alt_def})) or
  to the number of nodal domains of $f$ being equal to $\phi+1$.
\end{remark}

\begin{proof}
  We use induction on the number of internal zeros of $f$ to show that
  the eigenvalue is simple.  If $f$ has no internal zeros, then we
  know $f$ corresponds to the groundstate eigenvalue, which is
  simple.

  Now suppose $f$ has $\phi > 0$ internal zeros.  By way of
  contradiction, assume that $\lambda$ is not simple.  Choose an
  arbitrary zero $\zeta$ of $f$ and another eigenfunction $g$.  Cut
  $\Gamma$ at $\zeta$; making this cut will disconnect the graph into
  two subgraphs since $\zeta$ cannot lie on a cycle of $\Gamma$.  On at
  least one of these subgraphs, $g$ is non-zero and not a multiple of
  $f$ (otherwise, it cannot be a different eigenfunction).  We will
  now analyze the eigenfunctions on this subgraph $\Gamma'$.

  On the graph $\Gamma'$, $f$ and $g$ satisfy the same $\delta$-type
  conditions at all vertices except possibly the new leaf $\zeta$.  We
  denote by $\Gamma'_\tau$ as the graph $\Gamma'$ with the conditions
  $\Phi'(\zeta) = \tau\Phi(\zeta)$.  We know that $(\lambda,f)$ is an
  eigenpair on $\Gamma'_{\infty}$ and similarly, there exists $\alpha$
  such that $(\lambda,g)$ is an eigenpair on $\Gamma'_{\alpha}$.
  However, since $\Gamma'_{\infty}$ contains fewer internal zeros of
  $f$ than $\Gamma$ does, by the inductive hypothesis $\lambda$ is
  simple on $\Gamma'_{\infty}$ so $\alpha \neq \infty$.

  Observe that $f'(\zeta)$ is non-zero; if it was, the function $f$
  would be identically zero on the whole edge containing $\zeta$ and,
  therefore, at the end-vertices of the edge.  Thus, the
  inequalities in (\ref{eq:interlacing1}) with $\alpha' =
  \infty$ become strict and $\Gamma'_\alpha$ and
  $\Gamma'_{\infty}$ cannot have the same eigenvalue $\lambda$.

  Now we show that $f$ is eigenfunction number $n = \phi + 1$.  By
  Remark~\ref{rem:no_zero_cycle}  there are $\nu
  = \phi + 1$ nodal domains.  Since $\lambda$ is simple, we know from
  Theorem 5.2.8 of \cite{BK_graphs} that
 \begin{equation*}
   n - \beta \leq \nu \leq n\qquad\mbox{and}\qquad
   n \leq \phi + 1 \leq n + \beta
 \end{equation*}
 where $\beta$ is the Betti number of $\Gamma$.  However, since $\nu =
 \phi + 1$, both inequalities hold only if $\nu = n = \phi + 1$ so $f$
 is indeed the eigenfunction number $\phi + 1$.
\end{proof}

Below we only include the parts of the proof that differ from Theorem \ref{thm:main_cut}.

\begin{proof}[Proof of Theorem~\ref{thm:main_cut_2}]
  In the proof of Theorem~\ref{thm:main_cut}
  (section~\ref{sec:proof_main_cut_proof}), the fact that
  $H_{\boldsymbol{\gamma}}$ is an operator on a tree was used to show
  that its eigenvalue is simple and to find the sequence number of
  $\lambda$ in the spectrum.  Theorem~\ref{thm:simple} allows us to do
  the same in the graph with fewer cuts.

  Indeed, on the cut graph $\Gamma_{\widetilde{\boldsymbol\gamma}}$,
  $\psi$ is non-zero on all cycles and internal vertices and therefore
  by Theorem \ref{thm:simple}, the eigenvalue is simple and has number
  $\phi + 1$ in the spectrum of $H_{\boldsymbol{\widetilde{\gamma}}}$.
  Since the eigenvalue is simple, we can still apply Lemma
  \ref{lem:index}.  The rest of the proof goes through, with the
  amendment that the index of the critical point of $F_3$ is $n+\eta$,
  since we now have $\eta$ cuts instead of $\beta$ cuts.  Using
  equation (\ref{eq:eta}), we finally get that the Morse index of the
  critical point is
  \begin{equation*}
    (n+\eta)-(\phi-1) = n+(1+\phi-\nu)-\phi-1 = n - \nu.
  \end{equation*}
\end{proof}

\section*{Acknowledgment}
\label{sec:ack}

We are grateful to Y.~Colin~de~Verdi\`ere for numerous insightful
discussions and pointing out errors in earlier versions of the proof
of Theorem~\ref{thm:main_cut}.  The crucial idea that extending
$\alpha_j$ into the complex plane might be fruitful was suggested to
us by P.~Kuchment.  For this and many other helpful suggestions we are
extremely grateful.  We would also like to thank R.~Band, J.~Robbins,
and U.~Smilansky for encouragement and discussions and the anonymous
referees for numerous corrections.  GB was partially supported by the
NSF grant DMS-0907968.

\bibliographystyle{cj}
\bibliography{nodal_mag}

\end{document}